\documentclass[runningheads]{llncs}
\usepackage[utf8]{inputenc}

\newif\ifproc
\procfalse

\usepackage{graphicx}
\usepackage{hyperref}
\usepackage{amssymb, amsmath}
\graphicspath{{figures/}}

\renewcommand{\paragraph}[1]{\smallskip\noindent\textbf{\textsf{#1}}}

\usepackage{csquotes}
\usepackage{subcaption}

\usepackage{xcolor}
\usepackage{booktabs}
\usepackage{colortbl}
\colorlet{tablerowcolor}{gray!10}
\newcommand{\rowcol}{\rowcolor{tablerowcolor}}

\usepackage[inline]{enumitem}
\setlist[enumerate]{nosep}
\usepackage{hyperref}
\usepackage{timetravel}
\setCounters{theorem,lemma,corollary}
\usepackage{wrapfig}
\usepackage{multirow}
\usepackage[misc,geometry]{ifsym}

\DeclareMathOperator{\slope}{slope}
\newcommand{\eps}{\varepsilon}

\let\doendproof\endproof
\renewcommand\endproof{~\hfill$\qed$\doendproof}
\spnewtheorem*{proofWithFormula}{Proof}{\itshape}{\rmfamily}
\spnewtheorem*{sketchWithFormula}{Proof sketch}{\itshape}{\rmfamily}
\spnewtheorem*{sketch}{Proof sketch}{\itshape}{\rmfamily}

\newcounter{casecounter}
\newcounter{subcasecounter}
\newcounter{subsubcasecounter}
\makeatletter
\newcommand{\ccase}[1]{%
  \stepcounter{casecounter}%
  \setcounter{subcasecounter}{0}%
  \setcounter{subsubcasecounter}{0}%
  \protected@write \@auxout {}{\string \newlabel {#1}{{\thecasecounter}{\thepage}{\thecasecounter}{#1}{}} }%
  \hypertarget{#1}{\noindent\textbf{Case \thecasecounter.}}
}
\newcommand{\subcase}[1]{%
  \stepcounter{subcasecounter}%
  \setcounter{subsubcasecounter}{0}%
  \protected@write \@auxout {}{\string \newlabel {#1}{{\thecasecounter.\thesubcasecounter}{\thepage}{\thecasecounter.\thesubcasecounter}{#1}{}} }%
  \hypertarget{#1}{\noindent\textbf{Case \thecasecounter.\thesubcasecounter.}}
}
\newcommand{\subsucase}[1]{%
  \stepcounter{subsubcasecounter}%
  \protected@write \@auxout {}{\string \newlabel {#1}{{\thecasecounter.\thesubcasecounter.\thesubsubcasecounter}{\thepage}{\thecasecounter.\thesubcasecounter.\thesubsubcasecounter}{#1}{}} }%
  \hypertarget{#1}{\noindent\textbf{Case \thecasecounter.\thesubcasecounter.\thesubsubcasecounter.}}
}
\makeatother

\newcommand{\Tred}{\ensuremath{T'}}
\newcommand{\Trred}{\ensuremath{T''}}

\title{Drawing planar graphs with few segments on a polynomial grid\thanks{
  This work was initiated at the Workshop on Graph and Network Visualization
  2017.
  The work of P.~Kindermann was partially supported by DFG grant SCHU~2458/4-1.}}

\author{Philipp Kindermann\inst{1}\ifproc$^{\textrm{(\Letter)}}$\orcidID{0000-0001-5764-7719}\fi
\and Tamara Mchedlidze\inst{2}
\and Thomas~Schneck\inst{3}
\and Antonios~Symvonis\inst{4}}
\institute{%
  Universit\"at W\"urzburg, Germany,
  \email{philipp.kindermann@uni-wuerzburg.de}
  \and
  Karlsruhe Institute of Technology (KIT), Germany,
  \email{mched@iti.uka.de}
  \and
  Universit\"at T\"ubingen, Germany,
  \email{thomas.schneck@uni-tuebingen.de}
  \and
  National Technical University of Athens, Greece,
  \email{symvonis@math.ntua.gr}}

\authorrunning{P. Kindermann et al.}
\titlerunning{Drawing planar graphs with few segments on the grid}

\begin{document}
\maketitle

\begin{abstract}
The visual complexity of a graph drawing can be measured by the number of geometric objects used for the representation of its elements. 
In this paper, we study planar graph drawings where edges are represented by few segments.
In such a drawing, one segment may represent multiple edges forming a path. %
Drawings of planar graphs with few segments were intensively studied in the past years. 
However, the area requirements were only considered for limited subclasses of planar graphs.
In this paper, we show that trees have drawings with $3n/4-1$ segments and $n^2$ area, 
improving the previous result of $O(n^{3.58})$.  We also show that 3-connected planar graphs 
and biconnected outerplanar graphs have a drawing with $8n/3-O(1)$ and $3n/2-O(1)$ segments, 
respectively, and $O(n^3)$ area.
\end{abstract}

\section{Introduction}

The quality of a graph drawing can be assessed in a variety of ways: area, 
crossing number, bends, angular resolution, and many more.
All these measures have their justification, but in general it is challenging 
to optimize all of them in a single drawing. 
Recently, the \emph{visual complexity} was suggested as another quality 
measure for drawings~\cite{s-dgfa-jgaa15}. The visual complexity denotes the 
number of simple geometric entities used in the drawing. 

The visual complexity of a straight-line graph drawing can be formalized as the 
number of segments  formed by its edges, which we refer to as 
\emph{segment complexity}. Notice that edges constituting a single segment form 
a path in the graph. The idea of representing graphs with fewer segments complies
with the Gestalt principles of perception,
which are  rules for the organization of perceptual scenes 
introduced in the area of psychology in the 19th century~\cite{kmv-gpgd-gd15}. 
According to the law of continuation, 
the edges forming a segment may be easier grouped by our perception into a single 
entity. Therefore, drawing graphs with fewer segments may ease their perceptual processing.
A recent user study~\cite{kms-eaadf-jgaa18} suggests that  lowering the segment 
complexity may positively influence
aesthetics, depending on the background of the observer, as long as it does not
introduce unnecessarily sharp corners.
From the theoretical perspective, it is natural to ask for a drawing of a graph with the smallest segment complexity.
It is not surprising that it is ${\sf NP}$-hard  to determine whether a graph has a drawing with segment complexity $k$~\cite{dmnw-nmsdp-jgaa13}. However, we can still expect to prove bounds for certain graph classes.

Dujmovi{\'c} et al.~\cite{desw-dpgfs-cgta07} were the first to study drawings
with few segments and provided upper and lower bounds for several planar
graph classes. Since then, several new results have been provided~(\cite{dm-dptfs-cg19,hkms-dpgfg-jgaa18,ims-dpc3c-jgaa17,m-vgoto-phd16,mnbr-mscd3-jco13}, refer also to Table~\ref{tab:results}). 
These results shed only a little light on the area requirements of the drawings.  
In particular, in his thesis, Mondal~\cite{m-vgoto-phd16} gives an algorithm for triangulations
that produces drawings with $8n/3-O(1)$ segments on a grid of size $2^{O(n\log n )}$ in general and
 $2^{O(n)}$ for triangulations of bounded degree. Even with this large
grid, the algorithm uses
substantially more segments than the best-known algorithm for triangulations
without the grid requirement by
Durocher and Mondal~\cite{dm-dptfs-cg19}, which uses $7n/3-O(1)$ segments.
Recently, H\"ultenschmidt et al.~\cite{hkms-dpgfg-jgaa18} presented algorithms
that produce drawings with  $3n/4$ segments and $O(n^{3.58})$ area for trees, and
 $3n/2$ and $8n/3$ segments for outerplanar graphs and 3-trees, respectively, and $O(n^{3})$ area.    
 Igamberdiev et al.~\cite{ims-dpc3c-jgaa17} have provided an algorithm to construct 
drawings of planar cubic 3-connected graphs with $n/2$ segments and~$O(n^2)$ area.

\paragraph{Our Contribution.} 
In this paper, we concentrate on finding drawings with low segment complexity on a small grid.
Our contribution is summarized in Table~\ref{tab:results}.
In Section~\ref{sec:trees}, we show that every tree has a drawing with at most $3n/4-1$ segments on the $n\times n$
grid, improving the area bound by Hültenschmidt et al.~\cite{hkms-dpgfg-jgaa18}.
We then focus on drawing 3-connected planar graphs in Section~\ref{sec:triconnected}.
Using a combination of Schnyder realizers and orderly spanning trees, we show
that every 3-connected planar graph can be drawn  with 
$m-(n-4)/3\le (8n-14)/3$ segments on an $O(n)\times O(n^2)$ grid.
Finally, in Section~\ref{sec:others}, we use this result to draw on an~$O(n)\times O(n^2)$ grid 
maximal 4-connected graphs  with $5n/2-4$ segments, 
biconnected outerplanar graphs  with $(3n-3)/2$ segments, 
connected outerplanar graphs with $(7n-9)/4$ segments, 
and connected planar graphs with $(17n-38)/6$ segments.
All our proofs are constructive and yield algorithms to obtain such drawings
in $O(n)$ time.
As a side result, we also prove that the total number of leaves in every Schnyder 
realizer of a 3-connected planar graph is at most $2n+1$, which was 
only known for maximal planar graphs~\cite{b-wtr-icalp02,man-cdrsl-ijfcs05}.
 For the results on biconnected outerplanar 3- and 4-connected graphs, we use techniques that have been
used to construct monotone drawings; thus, as a side result, these drawings are also monotone\footnote{A path $P$ in a straight-line drawing of a graph is \emph{monotone} if there
	exists a line~$l$ such that the orthogonal projections of the vertices of~$P$ 
	on~$l$ appear along~$l$ in the order induced by~$P$. A drawing is monotone
	if there is a monotone path between every pair of vertices.}.

\vspace{+0.3cm}
We note that there are three trivial lower bounds for the segment complexity of a general graph~$G=(V,E)$ with~$n$ vertices and~$m$ edges:
\begin{enumerate*}[label=(\roman*)]
	\item $\vartheta/2$, where $\vartheta$ is the number of odd-degree vertices, 
	\item $\max_{v\in V} \lceil \deg(v)/2\rceil$, and
	\item $\lceil m/(n-1)\rceil$.
\end{enumerate*}
These trivial lower bounds are the same as for the slope number
of graphs~\cite{wc-dcgms-CJ94}, that is, the minimum number of slopes required
to draw all edges, and the slope number is upper
bounded by the number of segments required.

Relevant to segment complexity are the studies by Chaplick et al.~\cite{cflrvw-dgflf-GD16,cflrvw-cdgfl-WADS17} who consider drawings where all edges are to be covered by few lines (or planes);
the difference to our problem is that collinear segments are counted only once in their model.
In the same fashion, Kryven et al.~\cite{krw-dgfcf-CALDAM18} aim to cover all
edges by few circles (or spheres).

\setlength{\tabcolsep}{2.9pt}
\begin{table}[t]
\centering
\caption{Upper and lower bounds on the visual complexity of segment drawings.  
  Here, $n$ is the number of vertices, $m$ is the number of edges,~$\vartheta$ is the number of odd-degree vertices, and~$b$ is the number of maximal biconnected components. Constant-term
additions or subtractions have been omitted.
Entries marked by a \textbf{*} are monotone drawings.
\ifproc
FV corresponds to the full version~\cite{kmss-dpgfs-arxiv19}.\fi}
\label{tab:results}
\medskip
\begin{tabular}{>{\hspace{-\tabcolsep}\,}lrlrlcr<{\hspace{-\tabcolsep}}>{\hspace{-\tabcolsep}\,}rr<{\hspace{-\tabcolsep}}>{\hspace{-\tabcolsep}\,}c<{\hspace{-\tabcolsep}}>{\hspace{-\tabcolsep}\,}l<{\hspace{-\tabcolsep}}>{\hspace{-\tabcolsep}\,}r}
\toprule
Class & \multicolumn{4}{c}{Segments} && \multicolumn{6}{c}{Segments on the grid}\\[1ex]
\cmidrule{2-5}\cmidrule{7-12}
 & \multicolumn{2}{c}{Lower b.} & \multicolumn{2}{c}{Upper b.} && \multicolumn{2}{r}{Segments} & \multicolumn{3}{c}{Grid} & Ref.\\
\midrule
& & & & &&& $3n/4$ & $O(n^2)$&$\times$&$ O(n^{1.58})$ & \cite{hkms-dpgfg-jgaa18} \\

tree & $\vartheta/2$ & \cite{desw-dpgfs-cgta07}  &\bfseries $\vartheta/2$ & \cite{desw-dpgfs-cgta07} &&& \boldmath $3n/4$  & \boldmath $n$ & \boldmath $\times$ & \boldmath $n$ & Th.~\ref{thm:tree}\\

& & & & &&& $\vartheta/2$  & \multicolumn{3}{c}{quasipoly.}& \cite{hkms-dpgfg-jgaa18} \\

\rowcol max. outerplanar & $n$ & \cite{desw-dpgfs-cgta07} & $n$ & \cite{desw-dpgfs-cgta07} &&& $3n/2$  & $O(n)$&$\times$&$ O(n^2)$ & \cite{hkms-dpgfg-jgaa18}\\

& &  &  &&  &\textbf{*}& \boldmath $m-n/2$  & \boldmath $O(n)$ & \boldmath $\times$ & \boldmath $ O(n^2)$ & \ifproc FV \else Th. \ref{thm:outerbicon} \fi\\

 \multirow{-2}{*}{2-conn. outerplanar} & \multirow{-2}{*}{$n$} & \multirow{-2}{*}{\cite{desw-dpgfs-cgta07}}  &   &  &&\textbf{*}& \boldmath $3n/2$  & \boldmath $O(n)$ & \boldmath $\times$ & \boldmath $ O(n^2)$& \ifproc FV \else Cor. \ref{cor:outerbicon} \fi \\

\rowcol & & &  &  &&& \boldmath $3n/2+b$  & \boldmath $O(n)$ & \boldmath $\times$ & \boldmath $ O(n^2)$ & \ifproc FV \else Th.~\ref{thm:outer} \fi \\

\rowcol \multirow{-2}{*}{outerplanar} & \multirow{-2}{*}{$n$} & \multirow{-2}{*}{\cite{desw-dpgfs-cgta07}} &  &  &&& \boldmath $7n/4$  & \boldmath $O(n)$ & \boldmath $\times$ & \boldmath $ O(n^2)$ & \ifproc FV \else Cor.~\ref{cor:outer} \fi\\

2-tree & $3n/2$ & \cite{desw-dpgfs-cgta07} & $3n/2$ & \cite{desw-dpgfs-cgta07} &&&&& \\

\rowcol planar 3-tree & $2n$ & \cite{desw-dpgfs-cgta07} & $2n$ & \cite{desw-dpgfs-cgta07}&&& $8n/3$  & $O(n)$&$\times$&$ O(n^2)$ & \cite{hkms-dpgfg-jgaa18}\\

2-conn. planar & $2n$ & \cite{desw-dpgfs-cgta07} & $8n/3$ & \cite{dm-dptfs-cg19} &&& \multicolumn{4}{c}{\textbf{$\mathbf \rightarrow$ planar}}& \\

\rowcol & & & & &&\textbf{*}& \boldmath $m-n/3$  & \boldmath $O(n)$ & \boldmath $\times$ & \boldmath $ O(n^2)$ & Th. \ref{thm:3con}\\

\rowcol \multirow{-2}{*}{3-conn. planar } & \multirow{-2}{*}{$2n$} & \multirow{-2}{*}{\cite{desw-dpgfs-cgta07}} & \multirow{-2}{*}{$5n/2$} & \multirow{-2}{*}{\cite{desw-dpgfs-cgta07}} &&\textbf{*}& \boldmath $8n/3$ & \boldmath $O(n)$ & \boldmath $\times$ & \boldmath $ O(n^2)$ & Cor. \ref{cor:3con} \\

cubic 3-conn. planar  & $n/2$ & \cite{mnbr-mscd3-jco13} & $n/2$ & \cite{ims-dpc3c-jgaa17} &&& $n/2$  & $O(n)$&$\times$&$ O(n)$ & \cite{ims-dpc3c-jgaa17}\\

\rowcol triangulation & $2n$ & \cite{dm-dptfs-cg19} & $7n/3$ & \cite{dm-dptfs-cg19} &&\textbf{*}& \boldmath $8n/3$  & \boldmath $O(n)$ & \boldmath $\times$ & \boldmath $ O(n^2)$ & Cor. \ref{cor:3con}\\

4-conn. planar & $2n$ & \cite{dm-dptfs-cg19} & $21n/8$ & \cite{dm-dptfs-cg19} &&\textbf{*}& \multicolumn{4}{c}{\textbf{\boldmath $\rightarrow$ 3-conn.}} & \\

\rowcol 4-conn. triang. & $2n$ & \cite{dm-dptfs-cg19} & $9n/4$ & \cite{dm-dptfs-cg19} &&\textbf{*}& \boldmath $5n/2$  & \boldmath $O(n)$ & \boldmath $\times$ & \boldmath $ O(n^2)$ & \ifproc FV \else Th. \ref{thm:4con} \fi \\

& & & & && \multicolumn{2}{r}{\boldmath $17n/3-m$} & \boldmath $O(n)$ & \boldmath $\times$ & \boldmath $ O(n^2)$  & \ifproc FV \else Th. \ref{thm:planar} \fi \\

\multirow{-2}{*}{planar} & \multirow{-2}{*}{$2n$} & \multirow{-2}{*}{\cite{dm-dptfs-cg19}} & \multirow{-2}{*}{$8n/3$} & \multirow{-2}{*}{\cite{dm-dptfs-cg19}} &&& \boldmath $17n/6$  & \boldmath $O(n)$ & \boldmath $\times$ & \boldmath $ O(n^2)$ & \ifproc FV \else Cor. \ref{cor:planar} \fi \\
\bottomrule
\end{tabular}
\end{table}

\section{Trees}\label{sec:trees}

Let $T=(V,E)$ be a tree with $n$ vertices.
In this section, we describe an algorithm to draw~$T$ with at most $3n/4-1$
segments on an $n\times n$ grid in $O(n)$ time.

If $T$ consist only of vertices of degree 1 and 2, then 
it is a path and we can draw it with 1 segment and $n \times 1$ area.
So, we will assume that there is at least one vertex with higher degree. We choose such a vertex as
the root of~$T$. Denote the number of degree-2 vertices by $\beta$ and
the number of leaves by $\alpha$. In the first step, we create another tree $\Tred$
with $n-\beta$ vertices by contracting  all edges incident to a degree-2 vertex.
We say that a degree-2 vertex~$u$ \emph{belongs} to a vertex~$v$ if~$v$ is the
first descendent of~$u$ in~$T$ that has  degree greater than~2. 
Note that~$\Tred$ has the same number of leaves as~$T$.
In the next step, we remove all leaves from~$\Tred$ and obtain a tree~$\Trred$ with $n-\beta-\alpha$
vertices; see Fig.~\ref{fig:treeDeg2Contraction}.

\begin{figure}[t]
  \centering
	\subcaptionbox{$T$\label{fig:exampleTreeA}}{%
    \includegraphics[page=1]{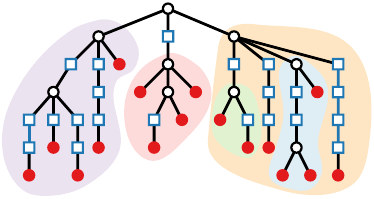}}
	\hfill
	\subcaptionbox{$\Tred$\label{fig:exampleTreeB}}{%
    \includegraphics[page=2]{figures/treeDrawing}}
	\hfill
	\subcaptionbox{$\Trred$\label{fig:exampleTreeC}}{%
    \includegraphics[page=3]{figures/treeDrawing}}
	\hfill
	\subcaptionbox{\label{fig:exampleTreeD}}{%
    \includegraphics[page=7]{figures/treeDrawing}}
  \caption{(a)~A tree $T$. Degree-2 vertices are squared, leaves are filled.
	(b)~The tree $\Tred$ obtained from $T$ by contracting the degree-2 vertices.
	(c)~The tree $\Trred$ obtained from $\Tred$ by removing all leaves.
  (d)~The drawing of our algorithm.}
  \label{fig:treeDeg2Contraction}
\end{figure}

The main idea of our algorithm is as follows. 
We draw~$\Trred$ with $n-\beta-\alpha-1$ segments.
Then, we add the $\alpha$ leaves in such a way that they either extend the 
segment of an edge, or that two of them share a segment, which results in at most $\alpha/2$
new segments. Finally, we place the~$\beta$ degree-2 vertices onto
the segments without increasing the number of segments. This way, we get a
drawing with at most $n-\beta-\alpha/2$ segments. Since~$\Tred$ has no degree-2 vertices,
more than half of its vertices are leaves, so $\alpha> (n-\beta)/2$.
Hence, the drawing has at most $3(n-\beta)/4< 3n/4$ segments.
Unfortunately, there are a few more details we have to take care of to achieve this bound.

Let~$v$ be a vertex in $\Trred$, and let $T[v]$ be the subtree of~$T$ rooted at~$v$.
Let~$n_v$ denote the number of vertices in $T[v]$. 
Let~$v_1,\ldots,v_k$ be the children of~$v$ in~$\Trred$. 
As induction hypotheses, we assume that each $T[v_i]$ is drawn inside a polygon~$B_i$
of \emph{dimensions} (edge lengths) $\ell_i,r_i,t_i,b_i, w_i,h_i$ as indicated in Fig.~\ref{fig:treeBox}
such that
\begin{enumerate}[nosep,label=$({\cal I}_{\arabic*})$]
	\item\label{invariant-box} no vertex of $T[v_i]$ lies to the top-left of~$v_i$, and
  \item\label{invariant-area} $B_i$ has area $n_i\times n_i$.
\end{enumerate}
 
Using three steps, we describe how to draw $T[v]$ inside a polygon $B_v$
of dimensions $\ell_v,r_v,t_v,b_v, w_v,h_v$ such that~$v$ lies at coordinate $(0,0)$. 
First, we place $T[v_1],\ldots,T[v_k]$. Second, we add the degree-2
vertices that belong to $v_1,\dots,v_k$. Finally, we add the leaf-children of $v$ and the degree-2 vertices belonging to them.

\noindent\textbf{Step 1.} We aim at placing $v_1$ directly below~$v$, and each polygon $B_{i},i\ge 2,$ to the right of
polygon~$B_{i-1}$, aligning~$v_i$ with the top boundary of~$B_{i-1}$;
see Fig.~\ref{fig:treeChildren}.
We place $v_1$ at coordinate $(0,-1-\sum_{i=1}^k t_i)$,
and each $v_i$ at coordinate $(x(v_{i-1})+r_{i-1}+\ell_i+1,y(v_{i-1})+t_{i-1})$,
where $x(v)$ and $y(v)$ are the $x$- and $y$-coordinates of~$v$, respectively.
 By invariant~\ref{invariant-area}, the total width and height of the drawings of $B_1,\dots,B_k$  are both at most $\sum_{i=1}^k n_{v_i}$. 

\noindent\textbf{Step 2.} Let~$\beta_i$ be the number of degree-2 vertices that belong to~$v_i$.
We move each polygon~$B_i$ downwards by~$\beta_i$, and place the degree-2 vertices
above~$v_i$; see Fig.~\ref{fig:treeChildrenDeg2}. This does not change the placement of any edge
of~$v$, the polygons are only moved downwards and are still disjoint, so the drawing remains planar.
The height of the drawing increases by at most
$\max_{i=1}^k\beta_i\le \sum_{i=1}^k\beta_i$
to $\sum_{i=1}^k (n_{v_i}+\beta_i)$, 
while the width remains $\sum_{i=1}^k n_{v_i}$.

\noindent\textbf{Step 3.}  
Let~$C_v$ the subtree of~$T[v]$ that consists of $v$, its leaf-children in~$\Tred$, and the degree-2 vertices belonging to them.
Let~$u_1,\ldots,u_a$ be the leaves of~$C_v$ and let~$\gamma_1,\ldots,\gamma_a$ be the number of degree-2 vertices that belong to them. Without lost of generality, assume that $\gamma_1\ge \ldots \ge \gamma_a$. 
We first consider the case where $a$ is even. We place the leaves alternatively to the bottom-left and to the top-right of~$v$ with as many rows between them and~$v$
as degree-2 vertices belong to them;
we draw each $u_{2i-1}$ and $u_{2i}$ on a segment through~$v$ with slope~$1/i$.
To this end, we place $u_{2i-1}$ at coordinate $(-(\gamma_{2i-1}+1)\cdot i,-\gamma_{2i-1}-1)$ %
and $u_{2i}$ at coordinate $((\gamma_{2i}+1)\cdot i,\gamma_{2i}+1)$ (recall that $u$ is placed at $(0,0)$).
We are able to place the degree-2 vertices that belong to these leaves between them and~$v$; see Fig.~\ref{fig:treeLeaves}.

If~$a$ is odd, then we apply the procedure described above for $u_1,\dots,u_{a-1}$. Vertex $u_a$ is placed as follows.
If~$v$ is a leaf in $\Trred$, then we place $u_a$ below~$v$ at coordinate $(0,-\gamma_a-1)$.
If~$v$ is not a leaf in $\Trred$, and no degree-2 vertex belongs to~$v$, and~$v$
is not the first child of its parent in~$\Trred$ (that is, there will be no edge that 
leaves~$v$ vertically above), then we place~$u_a$ above~$v$ at coordinate 
$(0,\gamma_a+1)$ such that it shares a segment with $(v,v_1)$. Otherwise, we place $u_a$ as every other vertex $u_i$ with odd index at coordinate $(-(\gamma_{a}+1)\cdot i,-\gamma_{a}-1)$. 

 By construction, the segments through $v$ drawn at step~3  cannot intersect $B_2,\ldots,B_k$, but there
might be an intersection between the segment from~$u_1$ to~$v$ and~$B_1$.
In this case, we move~$B_1$ downwards until the crossing disappears, which 
makes the drawing planar again. We call this action {\bf Step~4}. Thus, we have created a drawing of~$T[v]$
inside the polygon~$B_v$ that complies with invariant~\ref{invariant-box}. 
In the following,  we show that $B_v$ satisfies invariant~\ref{invariant-area}.

\begin{figure}[t]
  \centering
	\subcaptionbox{\label{fig:treeBox}}{%
    \includegraphics[page=1]{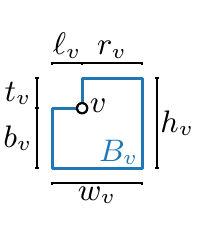}}
  \hfill
	\subcaptionbox{\label{fig:treeChildren}}{%
    \includegraphics[page=2]{treeBox}}
  \hfill
	\subcaptionbox{\label{fig:treeChildrenDeg2}}{%
    \includegraphics[page=3]{treeBox}}
  \hfill
	\subcaptionbox{\label{fig:treeLeaves}}{%
    \includegraphics[page=4]{treeBox}}
  \caption{Drawing of $T[v]$ with $k=4$. (a) $B_v$; (b) the children of~$v$ in $\Trred$;
    (c)~the degree-2 vertices belonging to these children;
    and (d) the remaining vertices of~$T[v]$ which form $C_v$.}
\end{figure}

  We analyze the width and height of the part of the
  drawing of $C_v$.
  Let  $\gamma^\mathrm{L}=\sum_{i=1}^{\lceil a/2\rceil} \gamma_{2i-1}$ and
  $\gamma^\mathrm{R}=\sum_{i=1}^{\lfloor a/2\rfloor} \gamma_{2i}$
  be the number of degree-2 vertices drawn to the left and right of~$v$, respectively,
  and let $\gamma=\gamma^\mathrm{L}+\gamma^\mathrm{R}$.
  
  Recall that $\gamma_1\ge\ldots\ge \gamma_a$ and leaf~$u_i$
  was placed at $y$-coordinate $\pm(\gamma_i+1)$.
  Hence, the vertices with the lowest and highest $y$-coordinate are~$u_1$ at 
  $y(u_1)=-\gamma_1-1$ and~$u_2$ 
  at $y(u_2)=\gamma_2+1$, respectively. Thus, the height of the drawing of~$C_v$ is 
  $1=1+a+\gamma$ if $a=0$; $2+\gamma_1=1+a+\gamma$ if $a=1$;
  and $3+\gamma_1+\gamma_2\le 1+a+\gamma$ if $a\ge 2$,
  so at most $1+a+\gamma$ in total.
  
  For analyzing the width of the drawing of~$C_v$,
  we first consider those vertices that are drawn to the right of~$v$.
  Let~$r$ be such that $u_{2r}$ is the rightmost vertex
  at $x$-coordinate $(\gamma_{2r}+1)\cdot r$.
  Since $\gamma_1\ge\ldots\ge \gamma_a$,
  we have that 
  \[ \gamma^\mathrm{R}=\sum_{i=1}^{\lfloor a/2\rfloor} \gamma_{2i}\ge \sum_{i=1}^{r} \gamma_{2i}
  \ge r \cdot \gamma_{2r}.\]
  Symmetrically, 
  let~$\ell$ be such that $u_{2\ell-1}$ is the leftmost vertex
  at $x$-coordinate $-(\gamma_{2\ell-1}+1)\cdot \ell$.
  We have that 
  \[\gamma^\mathrm{L}=\sum_{i=1}^{\lceil a/2\rceil} \gamma_{2i-1}\ge \sum_{i=1}^{\ell} \gamma_{2i-1}
  \ge \ell \cdot \gamma_{2\ell -1}.\]
  Hence, the total width of this part of the drawing is at most
  \[1+(\gamma_{2r}+1)\cdot r + (\gamma_{2\ell-1}+1)\cdot \ell\le 1+\ell+r+\gamma^\mathrm{L}+\gamma^\mathrm{R}\le 1+a+\gamma.\]

  Recall that before step~3 the width of the drawing of $T[v]$ was $\sum_{i=1}^k{n_{v_i}}$
  and the height was at most $\sum_{i=1}^k(n_{v_i}+\beta_i)$. 
  In step~3, the width increases by at most $1+a+\gamma$. 
  In step~4, we move the drawing of $T[v_1]$ downwards
  if it is crossed by the segment between $u_1$ and~$v$ until this crossing is
  resolved. There cannot be a crossing if $y(u_1)>y(v_1)$, 
  so we move it by at most $|y(u_1)|$ downwards, which is exactly the height of 
  the part of the drawing of~$C_v$ that lies below~$v$.
  Hence, the height in Steps~3 and 4  increases by at most the height of the drawing
  of~$C_v$, which is $1+a+\gamma$.
 Since $n_v=1+\sum_{i=1}^k (n_{v_i}+\beta_i)+a+\gamma$, the width and the height of $B_v$ is at most $n_v$.
  With this we complete the proof of invariant~\ref{invariant-area}.

We will now discuss the number of segments in $T$. Let~$r$ be the root of~$T$, and let~$v\in \Trred\setminus\{r\}$.
We need a few definitions;
see Fig.~\ref{fig:tree-segments}. Let $p_v$ be the parent of $v$
in $\Trred$. Let $P_v$ be the path between $v$ and $p_v$ in $T$;
let~$T^+[v] = T[v] \cup P_v$; let~$n_v^+$ be the number of vertices in~$T^+[v] \setminus \{p_v\}$;
let~$e_v$ be the edge of $P_v$ incident to $p_v$;
and let~$s_v$ be the number of segments used in the drawing of~$T^+[v]$.

\wormhole{treeSegments}
\newcommand{\treeSegmentsText}{%
  For any vertex~$v\neq r$ of~$\Trred$, if~$e_v$ is drawn vertical, then
  $s_v\le (3n^+_v-1)/4$, otherwise $s_v\le 3n^+_v/4$.}
\begin{lemma}\label{lem:tree-segments-internal}
  \treeSegmentsText
\end{lemma}
\begin{proof}
  We prove the lemma by induction on the height of~$\Trred$,
	so we can assume
  that the bound holds for all children of~$v$ in~$\Trred$. 
  Recall that $u_1,\ldots,u_a$ are the leaf-children of~$v$ in $\Tred$,
  $v_1,\ldots,v_k$ are the children of~$v$ in~$\Trred$,
  and~$v_1$ is connected to~$v$ by a vertical segment. 
  Let~$b$ be the number of degree-2 vertices that belong to~$v$. 
  Let~$n'=\sum_{i=1}^k n^+_{v_i}$; then, $n_v^+\ge n'+a+b+1$. 
  (There might be degree-2 vertices between~$v$ and its leaf-children in~$\Tred$
  which we do not count.)
  By induction, $s_{v_1}\le 3(n^+_{v_1}-1)/4$ and
  $s_{v_i}\le 3n^+_{v_i}/4$ for $2\le i\le k$, so
  $\sum_{i=1}^k s_{v_i}\le (3n'-1)/4$.
  It remains to analyze the number of segments for $C_v$ and for the path $P_v$.

  \begin{figure}[t]
  \centering
    \includegraphics[page=5]{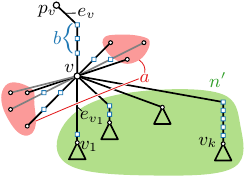}
    \caption{Illustration of $T^+[v]$ in the proof of Lemma~\ref{lem:tree-segments-internal}.}
    \label{fig:tree-segments}
  \end{figure}

  \widowpenalty0
  \clubpenalty0

  \ccase{c:leaf-nodeg2} $v$ is a leaf in~$\Trred$ and $b=0$. Then, $n^+_v\ge a+1$.
    Since~$v$ is a leaf in~$\Trred$, it has at least two children in~$T$,
    so $a\ge 2$.
  
  \subcase{sc:leaf-nodeg2-even} $a$ is even. 
    We use~$a/2$ for~$C_v$ plus one for the
    edge~$e_v$. Thus, 
    $s_v
    \le a/2+1
    \le (n^+_v-1)/2+1
    =(3n^+_v-n^+_v+2)/4
    \le (3n^+_v-1)/4$
    since $n^+_v\ge 3$.
  
  \subcase{sc:leaf-nodeg2-odd-vertical} $a$ is odd and~$e_v$ is vertical,
    so~$a\ge 3$ and $n^+_v\ge 4$. We use $(a-1)/2$ segments for $u_1,\ldots,u_{a-1}$
    and one segment for $u_a$ and~$e_v$. Thus,
    $s_v
    \le (a-1)/2+1
    \le n^+_v/2
    \le 3n^+_v/4-1$.
  
  \subcase{sc:leaf-nodeg2-odd-diagonal} $a$ is odd and~$e_v$ is not vertical.
    We use one more segment than in Case~\ref{sc:leaf-nodeg2-odd-vertical},
    so $s_v \le 3n^+_v/4$.
  
  \ccase{c:leaf-deg2} $v$ is a leaf in~$\Trred$ and $b>0$. Then, $a\ge 2$
    and $n^+_v\ge a+b+1\ge a+2\ge 4$
  
  \subcase{sc:leaf-deg2-even-vertical} $a$ is even and $e_v$ is vertical. 
    We use~$a/2$ segments for $u_1,\ldots,u_2$,
    and the degree-2 vertices that belong to~$v$ lie on a vertical segment with~$e_v$.
    Hence, we have $s_v\le a/2+1\le n^+_v/2\le 3n^+_v/4-1$.
  
  \subcase{sc:leaf-deg2-even-diagonal} $a$ is even and~$e_v$ is not vertical.
    We again have $n^+_v\ge 4$. The degree-2 vertices that belong to~$v$ now
    lie on a different segment than~$e_v$, so we have one more segment than
    in Case~\ref{sc:leaf-deg2-even-vertical}, so $s_v\le 3n/4$.
    
  \subcase{sc:leaf-deg2-odd} $a$ is odd. We have~$a\ge 3$ and thus $n^+_v\ge 5$. 
    We have drawn $u_1,\ldots,u_{a-1}$ paired up. We have drawn~$u_a$ on
    a vertical segment with the degree-2 vertices that belong
    to~$v$, and we have possibly one more segment for~$e_v$. Hence, we have
    $s_v\le (a-1)/2+2\le (n^+_v+1)/2
    = (3n_v^+ -n_v^+ +2)/4
    \le (3n_v^+ -3)/4$.
  
  \ccase{c:noleaf-nodeg2} $v$ is not a leaf in~$\Trred$ and $b=0$.
    We have $n^+_v\ge n'+a+1$, so $n'\le n^+_v-a-1$.
  
  \subcase{sc:noleaf-nodeg2-0-vertical} $a=0$ and~$e_v$ is vertical. 
    Then, $n_v^+=n'+1$  and~$e_v$ lies on a vertical segment with the edge $e_{v_1}$. 
    Hence, $s_v\le (3n'-1)/4=(3n^+_v-4)/4$.
  
  \subcase{sc:noleaf-nodeg2-0-diagonal} $a=0$ and~$e_v$ is not vertical. 
    Again, $n^+_v=n'+1$. We use one segment for~$e_v$,
    so we have $s_v\le (3n'-1)/4+1=3n^+_v/4$.
  
  \subcase{sc:noleaf-nodeg2-even} $a\ge 2$ is even. 
    We use $a/2$ segments for $C_v$ %
    and one more for~$e_v$.
    Hence, 
    $s_v
    \le (3n'-1)/4+a/2+1
    =(3n'+2a+3)/4
    \le (3n^+_v-a)/4
    \le (3n^+_v-2)/4$.
    
  \subcase{sc:noleaf-nodeg2-odd-vertical} $a$ is odd and~$e_v$ is vertical.
    We use $(a+1)/2$ segments for $u_1,\ldots,u_a$, but~$e_v$ shares its vertical
    segment with $e_{v_1}$. Hence, 
    $s_v
    \le (3n'-1)/4+(a+1)/2 
    = (3n'+2a+1)/4)
    \le (3n^+_v-a-2)/4
    \le (3n^+_v-3)/4$.
    
  \subcase{sc:noleaf-nodeg2-odd-diagonal} $a$ is odd and~$e_v$ is not vertical.
    In this case, we place~$u_a$ above~$v$ such that it lies on a segment with~$e_{v_1}$.
    We use $(a-1)/2$ segments
    for $u_1,\ldots,u_{a-1}$ and one segment for~$e_v$, so we have the same number of segments as in Case~\ref{sc:noleaf-nodeg2-odd-vertical}.
  
  \ccase{c:noleaf-deg2} $v$ is not a leaf in~$\Trred$ and $b>0$. 
    We have $n^+_v\ge n'+a+b+1\ge n'+a+2$.
    
  \subcase{sc:noleaf-deg2-even} $a$ is even.
    We use $a/2$ segments for $u_1,\ldots,u_a$.  The edges of the path $P_v$ %
    share a vertical segment with~$e_{v_1}$. We use at most one more
    segment for~$e_v$, so
    $s_v
    \le (3n'-1)/4+a/2+1
    =(3n'+2a+3)/4
    \le (3n^+_v-a-3)/4
    \le (3n^+_v-3)/4$.
    
  \subcase{sc:noleaf-deg2-odd-vertical} $a$ is odd and~$e_v$ is vertical.
    We use the exact same number of segments %
    as in Case~\ref{sc:noleaf-nodeg2-odd-vertical},
    so $s_v\le (3n^+_v-3)/4$.
    
  \subcase{sc:noleaf-deg2-odd-diagonal} $a$ is odd and~$e_v$ is not vertical.
    We use $(a+1)/2$ segments for $u_1,\ldots,u_a$.  The edges of the path $P_v$   %
    share a vertical segment with~$e_{v_1}$, and we need one more
    segment for~$e_v$. Hence,
    $s_v
    \le (3n'-1)/4+(a+1)/2+1
    =(3n'+2a+5)/4
    \le (3n^+_v-a-1)/4
    \le (3n^+_v-2)/4$.
\end{proof}

\widowpenalty10000
\clubpenalty10000

Now we can bound the total number of segments in the drawing of~$T$.

\begin{lemma}\label{lem:tree-segments}
  Our algorithm draws~$T$ with at most $3n/4-1$ segments if $n\ge 3$.
\end{lemma}
\begin{proof}
  If~$T$ is a path with $n\ge 3$, then the bound trivially holds. If~$T$ is a subdivision
  of a star, then the bound also clearly holds. Otherwise,~$\Trred$ consists of
  more than one vertex. 
  Let~$v_1,\ldots,v_k$ be the children of the root~$r$ of~$\Trred$ such that~$v_1$
  is connected by a vertical edge. Recall that ~$n'=\sum_{i=1}^k n^+_{v_i}$. 
  By Lemma~\ref{lem:tree-segments-internal}, the subtrees $T[v_i]^+$, $i=1,\dots, k$ contribute at most
  $(3n'-1)/4$ segments to the drawing of~$T$.
  Let~$a$ be the number of leaf children of~$r$ in~$\Tred$.
  If~$a$ is even, then we use~$a/2$ segments to draw them. If~$a$ is odd,
  then we align one of them with the vertical segment of $v_1$, and draw the
  remaining with $(a-1)/2$ segments. Since~$n\ge n'+a+1$, the total number
  of segments is at most $(3n'-1)/4+a/2\le 3n/4-a/4-1\le 3n/4-1$.
\end{proof}

All steps of the algorithm work in linear time. Sorting the 
leaf-children by the number of degree-2 vertices belonging to them
can also be done in linear time with, e.g., CountingSort, as the numbers
are bounded by~$n$. 
Thus, Theorem~\ref{thm:tree} follows.
Fig.~\ref{fig:exampleTreeD} shows the result of our algorithm for the tree 
of Fig.~\ref{fig:exampleTreeA}.

\begin{theorem}\label{thm:tree}
  Any tree with $n\ge 3$ vertices can be drawn planar on an $n\times n$ grid with
  $3n/4-1$ segments in $O(n)$ time.
\end{theorem}

\section{3-connected planar graphs}\label{sec:triconnected}

In this section, we present an algorithm to compute planar drawings 
with at most $(8n-14)/3$ segments for 3-connected planar graphs. 

\begin{figure}[t]
  \begin{minipage}[b]{.3\columnwidth}
    \centering
    \includegraphics{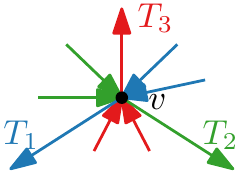}
    \caption{Edges in a Schnyder realizer.}
    \label{fig:schnyder-order}
  \end{minipage}
  \hfill
  \begin{minipage}[b]{.33\columnwidth}
    \centering
    \includegraphics{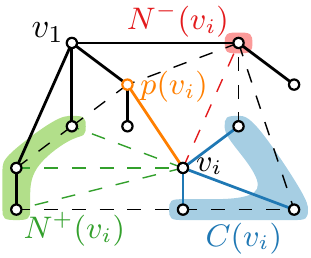}
    \caption{Definition of orderly spanning tree (bold).}
    \label{fig:orderly}
  \end{minipage}
  \hfill
  \begin{minipage}[b]{.3\columnwidth}
    \centering
    \includegraphics{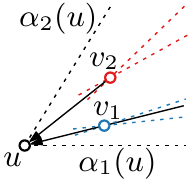}
    \caption{Definition of slope-disjointness.}
    \label{fig:slopedisjoint}
  \end{minipage}
\end{figure}

Let $G$ be a triangulation. Let $v_1,v_2,v_3$ be the vertices of the outer face.
We decompose the interior edges into three \emph{Schnyder trees} $T_1$, $T_2$, and $T_3$ rooted
at $v_1$, $v_2$, and $v_3$, respectively. The edges of the trees are oriented 
towards their roots. For $k\in\{1,2,3\}$, we call each edge in~$T_k$ a \emph{$k$-edge}
and the parent of a vertex in~$T_k$ its \emph{$k$-parent}.
The decomposition is a \emph{Schnyder realizer}~\cite{s-epgg-soda90} if at every 
interior vertex the edges are counter-clockwise ordered as: outgoing 1-edge, 
incoming 3-edges, outgoing 2-edge, incoming 1-edges, outgoing 3-edge, 
and incoming 2-edges; see Fig~\ref{fig:schnyder-order}.
A Schnyder tree~$T_k$ also contains the exterior edges of~$v_k$,
so each exterior edges lies in two Schnyder trees and each~$v_k$ is a leaf in
the other two Schnyder trees; hence, each Schnyder tree is a spanning tree.

For 3-connected planar graphs, Schnyder realizers also exist~\cite{DiBattista1999,FelsnerZ08}, but the interior edges
can be \emph{bidirected}: an edge $(u,v)$ is bidirected if it is an outgoing
$i$-edge at~$u$ and an outgoing $j$-edge at~$v$ with $i\neq j$. 
All other edges are \emph{unidirected}, that is, they are an outgoing $i$-edge
at~$u$ and an incoming $i$-edge at~$v$ (or vice-versa).
The restriction on the 
cyclic ordering around each vertex remains the same, but now the Schnyder
trees are not necessarily edge-disjoint.

Chiang et al.~\cite{cll-ostwa-sicomp05} have introduced the notion of 
\emph{orderly spanning trees}.
Recently, orderly spanning trees were redefined by Hossain and 
Rahman~\cite{hr-gstgd-tcs15} as \emph{good spanning trees}. We will use the
definition by Chiang et al., but note that these two definitions are equivalent. 
Two vertices in a rooted spanning tree are \emph{unrelated} if neither of them
is an ancestor of the other one. 
A tree is \emph{ordered} if the circular order of the edges around each vertex is fixed.
Let $G=(V,E)$ be a plane graph and let~$r\in V$ lie on the outer face.
Let~$T$ be an ordered spanning tree of~$G$ rooted at~$r$ 
that respects the embedding of~$G$.
Let $v_1,\ldots,v_n$ be the vertices of~$T$ as encountered in a counter-clockwise
pre-order traversal. For any vertex~$v_i$, let $p(v_i)$ be its parent in~$T$,
let $C(v_i)$ be the children of~$v$ in~$T$, let $N(v_i)$ be the neighbors
of~$v_i$ in $G$ that are unrelated to~$v_i$; see Fig.~\ref{fig:orderly}. Further, let $N^-(v_i)=\{v_j\in N(v_i)\mid j<i\}$
and $N^+(v_i)=\{v_j\in N(v_i) \mid j>i\}$. Then,~$T$ is called \emph{orderly} if
the neighbors around every vertex~$v_i$ are in counter-clockwise order
$p(v_i)$, $N^-(v_i)$, $C(v_i)$, $N^+(v_i)$. 
In particular, this means that there
is no edge in~$G$ between~$v_i$ and an ancestor in~$T$ that is not its parent and
there is no edge in~$G$ between~$v_i$ and a descendent in~$T$ that is not its child.
This fact is crucial, as it allows us to draw a path in an orderly spanning tree
on a single segment without introducing overlapping edges.

Angelini et al.~\cite{acbfp-mdg-12} have introduced the notion of a
\emph{slope-disjoint} drawing of a rooted tree~$T$, which is defined as follows;
see Fig.~\ref{fig:slopedisjoint}. 
\begin{enumerate}[label=(S\arabic*)]
	\item For every vertex $u$ in $T$, there exist two slopes~$\alpha_1(u)$ and~$\alpha_2(u)$ 
    with $0<\alpha_1(u)<\alpha_2(u)<\pi$, such that, for every edge~$e$ that is
    either $(p(u),u)$ or lies in~$T[u]$, it holds that $\alpha_1(u)<\slope(e)<\alpha_2(u)$;
  \item for every directed edge $(v,u)$ in~$T$, it holds that
    $\alpha_1(u)<\alpha_1(v)<\alpha_2(v)<\alpha_2(u)$
    (recall that edges are directed towards the root); and
  \item for every two vertices $u,v$ in $T$ with $p(u)=p(v)$, it holds that either
    $\alpha_1(u)<\alpha_2(u)<\alpha_1(v)<\alpha_2(v)$ or
    $\alpha_1(v)<\alpha_2(v)<\alpha_1(u)<\alpha_2(u)$.
\end{enumerate}

\begin{lemma}[\cite{acbfp-mdg-12}]\label{lem:slope-disjoint}
  Every slope-disjoint drawing of a tree is planar and monotone.
\end{lemma}

We will now create a special slope-disjoint drawing for rooted
orderly trees.

\wormhole{goodtreeDrawing}
\newcommand{\goodtreeDrawingText}{Let $T=(V,E)$ be an ordered tree rooted at a vertex~$r$ with~$\lambda$ leaves.
  Then, $T$ admits a slope-disjoint drawing with~$\lambda$ segments
  on an $O(n)\times O(n^2)$ grid such that all slopes are integer.
  Such a drawing can be found in $O(n)$ time.}
\begin{lemma}\label{lem:goodtree-drawing}
  \goodtreeDrawingText
\end{lemma}
\begin{sketchWithFormula}
  Let $v_1,\ldots,v_n{=}r$ be the vertices of $T$ as encountered in a counter-clockwise
  post-order traversal. Let $e_i=(v_i,p(v_i)),1\le i<n$.
  We assign the slopes to the edges of $T$ in the
  order $e_1,\ldots,e_{n-1}$.
  We start with assigning slope~$s_1=1$ to~$e_1$.
  For any other edge~$e_i,1<i<n$, if~$v_i$ is a leaf in~$T$, then we assign the slope
  $s_i=s_{i-1}+1$ to~$e_i$. Otherwise, since we traverse
  the vertices in a post-order, $p(v_{i-1})=v_i$ 
  and we assign the slope $s_i=s_{i-1}$ to $e_i$.
  
  We create a drawing~$\Gamma$ of~$T$ as follows. We place $r=v_n$
  at coordinate $(0,0)$. For every other vertex~$v$ with parent~$p$ that is
  drawn at coordinate $(x,y)$, we place~$v$ at coordinate $(x+1,y+\slope(v))$.
  
  We now analyze the number of segments used in~$\Gamma$;
  slope-disjointness, area, and running time are proven in 
  \ifproc
  the full version~\cite{kmss-dpgfs-arxiv19}.
  \else
  Appendix~\ref{app:triconnected}.
  \fi
  The root~$r$ is an endpoint of $\deg(r)$ segments and every leaf is an
  endpoint of exactly~1 segment. For every other vertex~$v$, its incoming edge
  and one of its outgoing edges lie on the same segment, so it is an endpoint
  of $\deg(v)-2$ segments. Since every segment has two endpoints, the total number
  of segments is 
  \begin{flalign*}
  &\frac12\left(\deg(r)+\sum_{v\text{ not leaf},v\neq r}(\deg(v)-2)+\sum_{v\text{ leaf}}\deg(v)\right)\\
  =&\frac12\left(\sum_{v}\deg(v)-2(n-\lambda-1)\right)=\frac12\left(2n-2-2n+2\lambda+2\right)=\lambda.\tag*{\qed}
  \end{flalign*}
\end{sketchWithFormula}

\begin{lemma}\label{lem:goodtree-to-3con}
  Let $G=(V,E)$ be a planar graph and
  let $T$ be an orderly spanning tree of~$G$ with~$\lambda$ leaves.
  Then, $G$ admits a planar monotone drawing with at most~$m-n+1+\lambda$ segments
  on an $O(n)\times O(n^2)$ grid in $O(n)$ time.
\end{lemma}
\begin{proof}
  We first create a drawing of~$T$ according to Lemma~\ref{lem:goodtree-drawing}.
  Now, we will plug this tree drawing into the algorithm by Hossain and 
  Rahman~\cite{hr-gstgd-tcs15}.
  
  This algorithm takes a slope-disjoint drawing of an
  orderly spanning tree~$T$ of $G$ and stretches the edges of~$T$ such that the
  remaining edges of~$G$ can be inserted without crossings. In this stretching
  operation, the slopes of the edges of $T$ are not changed. Further,
  the total width of the drawing only increases by a constant factor.
  Since~$T$ is drawn slope-disjoint, this produces a planar monotone drawing
  of~$G$ on an $O(n)\times O(n^2)$ grid. The algorithm runs in $O(n)$ time.
  
  To count the number of segments, assume that every edge of~$G$ that does
  not lie on~$T$ is drawn with its own segment. We have drawn~$T$ with~$\lambda$
  segments and the slopes of the edges of~$T$. Hence, our algorithm draws~$G$
  with~$\lambda$ segments for~$T$ and with $m-n+1$ segments for the
  remaining edges.
\end{proof}

Both Chiang et al.~\cite{cll-ostwa-sicomp05} and Hossain and 
Rahman~\cite{hr-gstgd-tcs15} have shown that every planar graph
has an embedding that admits an orderly spanning tree. However, we do not know 
anything about the number of leaves in an orderly spanning tree. 
Miura et al.~\cite{man-cdrsl-ijfcs05} have shown that Schnyder trees are 
orderly spanning trees, and it is known that every 3-connected planar graph
has a Schnyder realizer.

\begin{lemma}[\cite{man-cdrsl-ijfcs05}]\label{lem:schnyder-good}
  Let $G=(V,E)$ be a 3-connected planar graph and let $T_1$, $T_2$, and $T_3$
  be the Schnyder trees of a Schnyder realizer of~$G$. Then, $T_1$, $T_2$,
  and $T_3$ are orderly spanning trees of~$G$.
\end{lemma}

Bonichon et al.~\cite{b-wtr-icalp02} showed that there
is a Schnyder realizer for every triangulated graph such that the total number
of leaves in $T_1$, $T_2$, and $T_3$ is at most $2n+1$, which already gives us
a good bound on the number of segments for triangulations. We will now show
that the same holds for every Schnyder realizer of a 3-connected graph. 
Let~$v$ be a leaf in one of the Schnyder trees~$T_k$, $k\in\{1,2,3\}$, 
that is not the root of a Schnyder tree,
so~$v$ has no incoming $k$-edge.
Hence, the outgoing $(k+1)$-edge $(v,u)$ and the outgoing $(k-1)$-edge $(v,w)$ are
consecutive in the cyclical ordering around~$v$, so they lie on a common face~$f$.
We assign the pair $(v,k)$ to~$f$. 
We first show two lemmas.

\begin{lemma}\label{lem:schnyder-consecutive}
  Let~$u_1,\ldots,u_p$ be the vertices on an interior face~$f$ in ccw order.
  If $(u_1,k)$ and~$(u_2,i)$ are assigned to~$f$ for some $i,k\in\{1,2,3\}$, then $i=k$.
\end{lemma}
\begin{proof}
  Refer to Fig.~\ref{fig:schnyder-consecutive}. 
  By definition, $(u_1,u_2)$ is an outgoing $(k+1)$-edge at~$u_1$.
  Since~$u_1$ is a leaf in~$T_k$, $(u_1,u_2)$ cannot
  be an outgoing $k$-edge at~$u_2$. Hence, $(u_1,u_2)$ is either an incoming $(k+1)$-edge at~$u_2$
  (if it is unidirected), or an outgoing $(k-1)$-edge at~$u_2$ (if it is
  bidirected); it cannot be an outgoing $(k+1)$-edge since bidirected edges have to
  belong to two different Schnyder trees.  For $(u_2,i)$ to be assigned to~$f$,
  $u_2$ must have two outgoing edges at~$f$, so we are in the latter case.
  Hence,  $(u_2,u_3)$ is outgoing at~$u_2$, and
  by the cyclical ordering of the edges around~$u_2$, it is an outgoing
  $(k+1)$-edge. Thus,~$u_2$ has an outgoing $(k+1)$-edge and an outgoing $(k-1)$-edge
  at~$f$, so $i=k$.
\end{proof}

\begin{figure}[t]
  \centering
    \includegraphics[page=1]{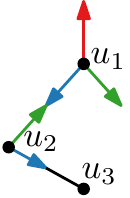}
    \quad\vline\quad
    \includegraphics[page=4]{schnyder-face}
    \includegraphics[page=1]{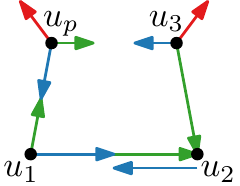}
    \includegraphics[page=2]{schnyder-face}
    \includegraphics[page=3]{schnyder-face}
  \caption{(Left) Proof of Lemma~\ref{lem:schnyder-consecutive} and (right) 
    proof of Lemma~\ref{lem:schnyder-between}.}
    \label{fig:schnyder-between}
    \label{fig:schnyder-consecutive}
\end{figure}

\wormhole{schnyderBetween}
\newcommand{\schnyderBetweenText}{%
  Let $u_1,u_2,\ldots,u_p$ be vertices on an interior face~$f$ in counter-clockwise order.  
  If $u_3,\ldots,u_p$ are assigned to~$f$, then neither~$u_1$ nor~$u_2$ are.
  }
\begin{lemma}\label{lem:schnyder-between}
  \schnyderBetweenText
\end{lemma}
\begin{sketch}
  From Lemma~\ref{lem:schnyder-consecutive}, it follows that
  $(u_3,k),\ldots,(u_p,k)$ are assigned to~$f$ for some~$k\in\{1,2,3\}$, so
  $(u_1,u_p)$ is an outgoing $(k+1)$-edge
  at~$u_p$ and $(u_2,u_3)$ is an outgoing $(k-1)$-edge at~$u_3$; since~$u_1$
  and~$u_p$ are leaves in~$T_k$,
  $(u_1,u_p)$ is either an incoming $(k+1)$-edge or an outgoing $(k-1)$-edge at~$u_1$
  and $(u_2,u_3)$ is either an incoming $(k-1)$-edge or an outgoing $(k+1)$-edge at~$u_2$.
  However, each of the four possible configurations violates the
  properties of a Schnyder realizer,
  as illustrated in Fig.~\ref{fig:schnyder-between}. 
  The full proof is given in 
  \ifproc
    the full version~\cite{kmss-dpgfs-arxiv19}.
  \else
    Appendix~\ref{app:triconnected}.
  \fi
\end{sketch}

Now we prove the bound on the number of leaves in a Schnyder realizer.

\begin{lemma}\label{lem:schnyder-leaves}
  Let~$T_1,T_2,T_3$ be a Schnyder realizer of a 3-connected planar graph~$G=(V,E)$.
  Then, there are at most $2n+1$ leaves in total in $T_1$, $T_2$, and~$T_3$.
\end{lemma}
\begin{proofWithFormula}
  Consider any interior face~$f$ of~$G$. By definition of the assignment, no vertex
  can be assigned to~$f$ twice.
  By Lemma~\ref{lem:schnyder-between}, at least two vertices
  on~$f$ are not assigned to~$f$, so we assign at most $\deg(f)-2$
  leaves to~$f$. At the outer face~$f^*$, every vertex that is not the root
  of a Schnyder tree can be assigned as a leaf at most once. However, the
  root of each of the Schnyder trees has no outgoing edges, but it can be a
  leaf in both the other two Schnyder trees.
  Hence, we assign at most $\deg(f^*)+3$ leaves to   the outer face.
  Let~$F$ be the faces in~$G$.
  Since, for every Schnyder tree, each of its leaves gets assigned to exactly one face,
  the total number of leaves in $T_1$, $T_2$, and~$T_3$ is at most
  \[\sum_{f\in F}\left(\deg(f)-2\right)+5= 2m-2|F|+5=2m+2n-2m-4+5=2n+1.\tag*{\qed}\]
\end{proofWithFormula}

Now we have the tools to prove the main result  of this section.

\begin{theorem}\label{thm:3con}
  Any 3-connected planar graph can be drawn planar monotone
  on an $O(n)\times O(n^2)$ grid with $m-(n-4)/3$ segments
  in~$O(n)$ time.
\end{theorem}
\begin{proof}
  Let $G=(V,E)$ be a 3-connected planar graph. We compute a Schnyder realizer
  of~$G$, which is possible in $O(n)$ time. 
  By Lemma~\ref{lem:schnyder-leaves}, the Schnyder trees have at most
  $2n+1$ leaves in total, so one of them, say~$T_1$, has at most $(2n+1)/3$
  leaves. By Lemma~\ref{lem:schnyder-good}, $T_1$ is an orderly spanning tree,
  so we can use Lemma~\ref{lem:goodtree-to-3con} to obtain a planar monotone
  drawing of~$G$ on an $O(n)\times O(n^2)$ grid with at most 
  $m-n+1+(2n+1)/3= m-n/3+4/3$ segments in $O(n)$ time.
\end{proof}

Since a planar graph has at most $m\le 3n-6$ edges, we have the following.

\begin{corollary}\label{cor:3con}
  Any 3-connected planar graph can be drawn planar monotone
  on an $O(n)\times O(n^2)$ grid with $(8n-14)/3$ segments
  in~$O(n)$ time.
\end{corollary}

\section{Other planar graph classes}\label{sec:others}

We can use the results of Section~\ref{sec:triconnected}
to obtain grid drawings with few segments for other planar graph classes on an $O(n)\times O(n^2)$
grid in $O(n)$ time. In particular, 
we can draw
\begin{enumerate*}[label=(\roman*)]
	\item 4-connected triangulations with $5n/2 - 4$ segments;
  \item biconnected outerplanar graphs with $m - (n - 3)/2\le (3n-3)/2$ segments;
  \item outerplanar graphs with $(7n-9)/4$ segments,
    or with $(3n - 5)/2 + b$ segments, where $b$ is its number of maximal biconnected components; and
  \item planar graphs with $(17n-38)/3-m$ or $(17n-38)/6$ segments.
\end{enumerate*} 
Details are given in 
\ifproc
  the full version~\cite{kmss-dpgfs-arxiv19}.
\else
  Appendix~\ref{app:others}.
\fi

\subsubsection*{Acknowledgements.}
We thank Roman Prutkin for the initial discussion of the problem and Therese Biedl for helpful comments.

\clearpage

\bibliographystyle{splncs04}
\bibliography{abbrv,fewsegments}

\ifproc
\end{document}
\fi

\clearpage
\appendix

\section{Omitted proofs from Section~\ref{sec:triconnected}}\label{app:triconnected}

\begin{backInTime}{goodtreeDrawing}

\begin{lemma}
  \goodtreeDrawingText
\end{lemma}
\begin{proof}
  It remains to show that~$\Gamma$ is slope-disjoint, the area is $O(n)\times O(n^2)$,
  and the algorithm runs in $O(n)$ time.
  
  We now prove that~$\Gamma$ is slope-disjoint.
  Let~$\eps>0$ be arbitrary small, $1/n$ should suffice.
  We set $\alpha_1(r)=1-\varepsilon$ and $\alpha_2(r)=s_n+\varepsilon$. 
  We assign the remaining values of $\alpha_1$ and~$\alpha_2$
  in pre-order. Let~$u$ be a vertex that has already been handled, that is,
  (S1) holds for~$u$, (S2) holds for $(p(u),u)$, and (S3) holds for $u$ and
  all of its siblings. Obviously, this holds for~$r$ in the beginning.
  Let $u_1,\ldots,u_k$ be the children of~$u$ in counter-clockwise order. By construction,
  we have $\slope(u_k)=\slope(u)<\alpha_2(u)$ and $\slope(u_k)>\alpha_1(u)$.
  Furthermore, by construction, all slopes in $T[u_i], 1<i\le k$ are
  larger than $\slope(u_{i-1})$ and at most $\slope(u_i)$. Hence, choosing
  $\alpha_1(u_1)=\alpha_1(u)+\eps$, 
  $\alpha_1(u_i)=\slope(u_{i-1})+\eps, 1<i\le k$, and
  $\alpha_2(u_i)=\slope(u_i)+\eps, 1\le i\le k$ satisfies (S1) for every
  $u_i$, (S2) for every edge $(u,u_i)$, and (S3) for every pair $(u_i,u_j)$.
  Repeating this construction establishes the conditions for every vertex,
  every edge, and every pair of siblings, so~$\Gamma$ is slope-disjoint.
  
  Finally, since every vertex is placed one $x$-coordinate to the right of its
  parent, we use at most~$n$ columns, and since the highest slope is~$n-1$,
  the drawing lies on a grid of size $O(n)\times O(n^2)$. Our algorithm consists
  of doing one post-order traversal and then placing the vertices, so it
  clearly takes $O(n)$ time.
\end{proof}
\end{backInTime}

\begin{backInTime}{schnyderBetween}
  \begin{lemma}
  \schnyderBetweenText
\end{lemma}
\begin{proof}
  From Lemma~\ref{lem:schnyder-consecutive}, it follows that
  $(u_3,k),\ldots,(u_p,k)$ are assigned to~$f$ for some~$k\in\{1,2,3\}$, so
  $(u_1,u_p)$ is an outgoing $(k+1)$-edge
  at~$u_p$ and $(u_2,u_3)$ is an outgoing $(k-1)$-edge at~$u_3$; since~$u_1$
  and~$u_p$ are leaves in~$T_k$,
  $(u_1,u_p)$ is either an incoming $(k+1)$-edge or an outgoing $(k-1)$-edge at~$u_1$
  and $(u_2,u_3)$ is either an incoming $(k-1)$-edge or an outgoing $(k+1)$-edge at~$u_2$.
  
  \textbf{Case 1:} $(u_1,u_p)$ is an incoming $(k+1)$-edge at~$u_1$
  and $(u_2,u_3)$ is an incoming $(k-1)$-edge at~$u_2$; see Fig.~\ref{fig:schnyder-between-none}. Then, neither~$u_1$
  nor~$u_2$ has two outgoing edges at~$f$, so there can be no pairs $(u,i)$
  and $(v,j)$ assigned to~$f$.
  
  \textbf{Case 2:} $(u_1,u_p)$ is an outgoing $(k-1)$-edge at~$u_1$
  and $(u_2,u_3)$ is an incoming $(k-1)$-edge at~$u_2$; see Fig.~\ref{fig:schnyder-between-u1}.
  By Lemma~\ref{lem:schnyder-consecutive}, if~$(u,i)$ is assigned to~$f$, then
  $i=k$; hence, $(u_1,u_2)$ has to be an outgoing $(k+1)$-edge at~$u_1$.
  By the cyclical ordering around~$u_2$, $(u_1,u_2)$ has to be either an incoming
  $(k-1)$-edge, or an outgoing $(k+1)$-edge at~$u_2$; both cases cannot be
  combined with an outgoing $(k+1)$-edge at~$u_1$. Hence, this case is not possible.
  
  \textbf{Case 3:} $(u_1,u_p)$ is an incoming $(k+1)$-edge at~$u_1$
  and $(u_2,u_3)$ is an outgoing $(k+1)$-edge at~$u_2$; see Fig.~\ref{fig:schnyder-between-u2}.
  This case is symmetric to Case 2: $(u_1,u_2)$ has to be an outgoing
  $(k-1)$-edge at~$u_2$ and either an outgoing $(k-1)$-edge or an incoming
  $(k+1)$-edge at~$u_1$, which is not possible.
  
  \textbf{Case 4:} $(u_1,u_p)$ is an outgoing $(k-1)$-edge at~$u_1$
  and $(u_2,u_3)$ is an outgoing $(k+1)$-edge at~$u_2$; see Fig.~\ref{fig:schnyder-between-u12}.
  By the same arguments as in Cases~2 and~3, $(u_1,u_2)$ has to be an outgoing
  $(k+1)$-edge at~$u_1$ and an outgoing $(k-1)$-edge at~$u_2$. However,
  since $(u_3,k),\ldots,(u_p,k)$ are assigned to~$f$, every edge $(u_q,u_{q+1}), 1\le q\le p$
  has to be an outgoing $(k+1)$-edge at~$u_q$, so there is a directed cycle
  in~$T_{k+1}$; a contradiction to $T_{k+1}$ being a tree. Thus, this case
  is also not possible.
\end{proof}
\end{backInTime}

\begin{figure}[t]
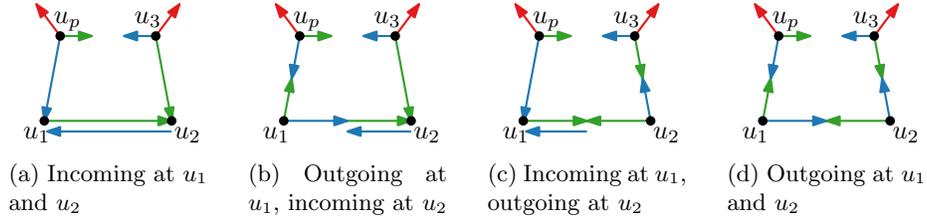

  \subcaptionbox{Incoming at $u_1$ and $u_2$\label{fig:schnyder-between-none}}{~\includegraphics[page=4]{schnyder-face}~}
  \hfill
  \subcaptionbox{Outgoing at $u_1$, incoming at $u_2$\label{fig:schnyder-between-u1}}{~\includegraphics[page=1]{schnyder-face}~}
  \hfill
  \subcaptionbox{Incoming at $u_1$, outgoing at $u_2$\label{fig:schnyder-between-u2}}{~\includegraphics[page=2]{schnyder-face}~}
  \hfill
  \subcaptionbox{Outgoing at $u_1$ and $u_2$\label{fig:schnyder-between-u12}}{~\includegraphics[page=3]{schnyder-face}~}
  \caption{Illustrations for the proof of Lemma~\ref{lem:schnyder-between}.}
  \label{fig:schnyder-between-app}
\end{figure}

\section{Other planar graph classes}\label{app:others}

In this section, we use the results of Section~\ref{sec:triconnected}
to obtain grid drawings with few segments for 4-connected triangulations,
(biconnected) outerplanar graphs, and planar graphs.
Using regular edge labelings, Zhang and He~\cite{zh-cotta-dcg05}
proved that any 4-connected triangulation admits a Schnyder tree with at most 
$\lceil(n+1)/2\rceil$ leaves. Applying this to Lemma~\ref{lem:goodtree-to-3con},
we find that our algorithm uses at most 
$m-n+1+\lceil(n+1)/2\rceil\le 3n-6-n+1+n/2+1=5n/2-4$ segments.

\begin{theorem}\label{thm:4con}
  Any 4-connected triangulation can be drawn planar monotone
  on an $O(n)\times O(n^2)$ grid with $5n/2-4$ segments
  in~$O(n)$ time.
\end{theorem}

We now consider outerplanar graphs.

\wormhole{outerbicon}
\newcommand{\outerbiconText}{%
  Any biconnected outerplanar graph can be drawn planar monotone
  on an $O(n)\times O(n^2)$ grid with $m-(n-3)/2$ segments
  in~$O(n)$ time.}
\begin{theorem}\label{thm:outerbicon}
  \outerbiconText
\end{theorem}
\begin{proof}
  Let $G=(V,E)$ be a biconnected outerplanar graph. We construct a 
  3-connected planar graph~$G'=(V',E')$ by adding a vertex~$r$ and connecting
  it to all vertices of~$V$. Hence, we have $n'=|V'|=n+1$ and $m'=|E'|=m+n$.
  Then, we compute a Schnyder realizer of~$G'$ such that one of its Schnyder
  trees, say~$T_3$, is rooted at~$r$. Since~$r$ is connected to all other vertices
  of~$V'$, $T_3$ consists of exactly those edges, so it has~$n'-1$ leaves.
  By Lemma~\ref{lem:schnyder-leaves}, at least one of~$T_1$ and~$T_2$ has
  $\lambda'\le n'/2+1$ leaves. We use Lemma~\ref{lem:goodtree-to-3con} to 
  obtain a planar monotone drawing of~$G'$ on an $O(n)\times O(n^2)$ grid with at most
  \[m'-n'+1+\lambda'\le m'-n'+1+\frac{n'}{2}+1=m'-\frac{n'}{2}+2=m+n-\frac{n+1}{2}+2=m+\frac{n}{2}+\frac{3}{2}\]
  segments in~$O(n)$ time. Since all edges of $E'\setminus E$ lie in~$T_3$, we can remove~$r$
  and those edges without splitting any segment into two segments. Hence, we
  obtain a drawing of~$G$ with 
  \[m+\frac{n}{2}+\frac{3}{2}-(m'-m)=m+\frac{n}{2}+\frac{3}{2}-n=m-\frac{n}{2}+\frac{3}{2}\]
  segments. Further, the monotonicity of~$G'$ depends only on the edges of its
  orderly spanning tree, so~$G$ is also drawn monotone.
\end{proof}

Outerplanar graphs have at most $m\le 2n-3$ edges, so this Corollary follows.

\begin{corollary}\label{cor:outerbicon}
  Any biconnected outerplanar graph can be drawn planar monotone
  on an $O(n)\times O(n^2)$ grid with $(3n-3)/2$ segments
  in~$O(n)$ time.
\end{corollary}

\wormhole{outer}
\newcommand{\outerText}{%
  Any connected outerplanar graph can be drawn planar 
  on an $O(n)\times O(n^2)$ grid with $(3n-5)/2+b$ segments,
  where~$b$ is its number of maximal biconnected components,
  in~$O(n)$ time.}
\begin{theorem}\label{thm:outer}
  \outerText
\end{theorem}
\begin{proof}
  Let $G=(V,E)$ be a connected outerplanar graph. We first augment~$G$
  to a biconnected outerplanar graph~$G'=(V',E')$ by adding the minimum number of edges
  required. Garc\'\i a et al.~\cite{ghnt-acog-algo98} gave an algorithm to
  do this in $O(n)$ time, and Read~\cite{r-nmdpg-cn87} has shown that, 
  if~$G$ consists of~$b$ maximal biconnected components, then the number of 
  edges required is at most $b-1$.
  Let~$d\le b-1$ be the number of edges added by this algorithm.
  Hence, we have~$n'=|V'|=n$ and $m'=|E'|= m+d$.
  We use Theorem~\ref{thm:outerbicon} to obtain a planar drawing of~$G'$
  on an $O(n)\times O(n^2)$ grid with at most $m'-n/2+3/2$
  segments.
  Removing the~$d$ added edges from~$G'$ splits at most~$d$ segments into
  two, so we obtain a drawing of~$G$ with at most
  $m'+d-n/2+3/2$ segments. 
  
  Since~$G'$ is outerplanar, we have $m'\le 2n-3$. Hence, the number of segments
  is at most $2n-3+d-n/2+3/2\le %
  3n/2+b-5/2$.
\end{proof}

By a simple case analysis, we can give a bound on the number of segments only
in terms of~$n$. If $m\le 7n/4-9/4$, then we draw~$G$ with~$m$ segments;
otherwise, $b\le 2n-3-m+1\le n/4+1/4$ and we use Theorem~\ref{thm:outer}
to draw~$G$ with at most $3n/2+b-5/2\le 3n/2+n/4+1/4-5/2=7n/4-9/4$
segments.

\begin{corollary}\label{cor:outer}
  Any connected outerplanar graph can be drawn planar 
  on an $O(n)\times O(n^2)$ grid with $(7n-9)/4$ segments
  in~$O(n)$ time.
\end{corollary}

Finally, we consider drawings of planar graphs.

\wormhole{planar}
\newcommand{\planarText}{%
  Any planar graph can be drawn planar 
  on an $O(n)\times O(n^2)$ grid with $(17n-38)/3-m$ segments
  in~$O(n)$ time.}
\begin{theorem}\label{thm:planar}
  \planarText
\end{theorem}
\begin{proof}
  Let $G=(V,E)$ be a planar graph. 
  We use a similar technique as for outerplanar graphs. We first augment~$G$ 
  to a 3-connected planar graph~$G'=(V',E')$ with $n'=|V'|=n$ and $m'=|E'|=m+d$ by adding edges.
  Then, we use Theorem~\ref{thm:3con} to draw~$G'$
  on an $O(n)\times O(n^2)$ grid with at most $m'-n/3-2/3$ segments in $O(n)$
  time. Finally, we remove the~$d$ added edges to obtain a drawing
  with at most $m'-n/3-2/3+d=m-n/3-2/3+2d$ segments.
  
  Unfortunately, adding a minimum number of edges to a planar graph to make
  it 3-connected is NP-hard~\cite{k-adpg-phd93} and we are not aware of any
  good bounds on the number of edges required. In the worst case, $G'$
  is a triangulation, so we have $m'\le 3n-6$ and $d\le 3n-6-m$. Hence, our drawing
  uses at most
  $m-n/3-2/3+2d\le m-n/3-2/3+6n-12-2m=17n/3-m-38/3$ segments.
\end{proof}

We can again use a simple case distinction to give a bound on the number of
segments purely in terms of~$n$. 
If~$m\le 17n/6-38/6$, then we draw~$G$ with~$m$
segments. Otherwise, we use Theorem~\ref{thm:planar} to draw~$G$ with at most
$17n/3-m-38/3\le 17n/6-38/6$ segments.

\begin{corollary}\label{cor:planar}
  Any planar graph can be drawn planar 
  on an $O(n)\times O(n^2)$ grid with $(17n-38)/6$ segments
  in~$O(n)$ time.
\end{corollary}

\end{document}